\documentclass[11pt,reqno]{article}

\usepackage[margin=1.5in, paperwidth=8.5in, paperheight=11in]{geometry}
\usepackage{amsfonts, amsmath, amssymb}
\usepackage{amsthm}
\usepackage{graphicx}
\usepackage{xcolor}
\usepackage[hyphens,spaces,obeyspaces]{url}
\usepackage{float}
\usepackage{physics}
\usepackage{bbm}
\usepackage{underscore}
\usepackage{array}
\usepackage{longtable}
\usepackage{caption}
\usepackage{qcircuit}
\usepackage{xy}
\usepackage{color}
\definecolor{darkblue}{rgb}{0,0,0.5}
\definecolor{darkgreen}{rgb}{0,0.5,0}
\usepackage[
     colorlinks=true,
     linkcolor=darkblue,
     urlcolor=darkblue,
     citecolor=darkgreen
]{hyperref}

\usepackage[capitalise]{cleveref}

\newcolumntype{L}[1]{>{\raggedright\let\newline\\\arraybackslash\hspace{0pt}}m{#1}}
\newcolumntype{C}[1]{>{\centering\let\newline\\\arraybackslash\hspace{0pt}}m{#1}}
\newcolumntype{R}[1]{>{\raggedleft\let\newline\\\arraybackslash\hspace{0pt}}m{#1}}

\newtheorem{lma}{Lemma}[section]

\newtheorem{thm}[lma]{Theorem}

\newtheorem{cor}[lma]{Corollary}

\theoremstyle{definition}

\theoremstyle{definition}
\newtheorem{dfn}[lma]{Definition}

\newtheorem{rmk}[lma]{Remark}

\crefname{lma}{Lemma}{Lemmas}
\crefname{thm}{Theorem}{Theorems}
\crefname{prop}{Proposition}{Propositions}
\crefname{exmp}{Example}{Examples}
\crefname{dfn}{Definition}{Definitions}
\crefname{pos}{Postulate}{Postulates}

\newcommand{\C}{\mathbb{C}}

\newcommand{\R}{\mathbb{R}}

\newcommand{\U}{\mathcal{U}}
\newcommand{\Sym}{\text{Sym}}

\newcommand{\D}{\mathcal{D}}
\newcommand{\Ha}{\mathcal{H}}

\newcommand{\ra}{\rightarrow}
\newcommand{\id}{\text{id}}
\newcommand{\Lc}{\mathcal{L}}

\newcommand{\one}{\mathbbm{1}}

\newcommand{\Eb}{\mathop{{}\mathbb{E}}}
\newcommand{\Hom}{\text{Hom}}

\begin{document}

\title{Quantum Private Broadcasting}

\author{
Anne Broadbent \thanks{University of Ottawa, Ottawa, Canada. \{\texttt{abroadbe,cschu059\}@uottawa.ca} }
\and Carlos E.~Gonz\'alez-Guill\'en \thanks{
Universidad Polit\'ecnica de Madrid, Madrid, Spain. \texttt{carlos.gguillen@upm.es} }
\and Christine Schuknecht \footnotemark[1]
 }

\date{}
\maketitle
\begin{abstract}

In \emph{Private Broadcasting}, a single plaintext is  \emph{broadcast} to multiple recipients in an encrypted form, such that each recipient can decrypt locally.   When the message is classical, a straightforward solution is to encrypt the plaintext with a single key shared among all parties, and to send to each recipient a  \emph{copy} of the ciphertext. Surprisingly, the analogous method is insufficient in the case where the message is quantum (\emph{i.e.}~in  \emph{Quantum Private Broadcasting (QPB)}). In this work, we give three solutions to $t$-recipient Quantum Private Broadcasting ($t$-QPB) and compare them in terms of key lengths.  The first method is the \emph{independent} encryption with the  quantum one-time pad, which requires a key linear in the number of recipients, $t$. We show that the key length can be decreased to be logarithmic in $t$ by using unitary $t$-designs. Our main contribution is to show that this can be improved to a key length that is logarithmic in the dimension of the symmetric subspace, using a new concept that we define of \emph{symmetric unitary $t$-designs}, that may be of independent interest.
\end{abstract}

\section{Introduction}

 How can we securely and efficiently broadcast a single message to $t$ recipients, such that the message is information-theoretically secure? In the case of an $n$-bit message~$m$, this can be achieved by the use of the one-time pad, where a \emph{single} key $k \in \{0,1\}^n$ is initially sampled uniformly at random, and distributed to all parties. Each party then receives the ciphertext
  $m \oplus k$, and can decrypt using their knowledge of~$k$; the cost of this method (in terms of key length) is therefore independent of~$t$.  In the case where the plaintext is a \emph{quantum} message, the situation is much more complex and counterintuitive. In fact, as we show, the use of the quantum one-time pad in an analogous fashion to above is in general not secure.

\subsection{Summary of Results}

We  formally define the $t$-recipient Quantum Private Broadcasting ($t$-QPB) problem in the information-theoretic setting, and show three methods to achieve it. Along the way, we define a new notion of \emph{designs}, applicable to the scenario where the input is in the symmetric subspace; we call these \emph{symmetric unitary $t$-designs} and note that these may be of independent interest beyond~$t$-QPB.

The first straightforward solution to the $t$-QPB problem is the encryption of each copy of the plaintext with the quantum one-time pad, using independent keys. This requires a key of length linear in $t$, the number of recipients, and is secure even if the adversary holds quantum side-information about the plaintext.  We observe however that this solution does not make use of the full structure of the problem, namely that each recipient receives the \emph{same} plaintext.

A unitary $t$-design~\cite{DCEL09,RS09} is a  finite set of unitary matrices, together with  a probability distribution, such that averaging up to $t$ copies of the same unitary over the design is equivalent to integrating up to $t$ copies of the same unitary over the whole unitary group with respect to the Haar measure. Intuitively, since encryption can be achieved with a Haar random unitary, it follows that unitary $t$-designs are a $t$-QPB scheme. Since the key length required for unitary $t$-designs is logarithmic in~$t$, this offers an exponential improvement in key length compared to the first solution.  Moreover, we show that unitary $t$-designs are secure against quantum side information, as long as the state to be encrypted is in the symmetric subspace. We can ensure the input state is always in the symmetric subspace by implementing a pre-broadcasting stage. This projects the state into the symmetric subspace, aborting the encryption protocol if the resulting state is not symmetric.

Our final solution takes full advantage of the structure of the $t$-QPB problem, and defines \emph{symmetric unitary $t$-designs} as a relaxation of unitary $t$-designs that mimic the action of the Haar measure \emph{on the symmetric subspace}. We show that, up to some reasonable assumptions, these are necessary and sufficient as $t$-QPB schemes and that they yield a key length logarithmic in $d_{\Sym}$ (the dimension of the symmetric subspace); this is still logarithmic in $t$, but with a smaller constant than the key length of encryption schemes derived from unitary $t$-designs. We also provide lower and upper bounds for both exact and approximate symmetric unitary $t$-designs with respect to~$d_{\Sym}$. This lower bound of $d_{\Sym}^2$ for exact symmetric unitary $t$-designs corresponds to the number of unitaries needed to perform the quantum one-time pad in the symmetric subspace, which is the $t$-QPB problem without the local decryption requirement.

We use the bounds for the size of weighted unitary $t$-designs as proven in~\cite{RS09} to compare the key length of a design as opposed to $t$ uses of the quantum one-time pad (QOTP). We compare the results for the qubit case in \cref{3key20},  which shows that when $t>5$, symmetric designs are a better choice than the QOTP, while it takes until $t>6$ for regular designs to be better than the QOTP. (The data for \cref{3key20} is given in \cref{sec:appendix-data}.)

\begin{figure}[H]
\begin{center}

\includegraphics[scale=.83]{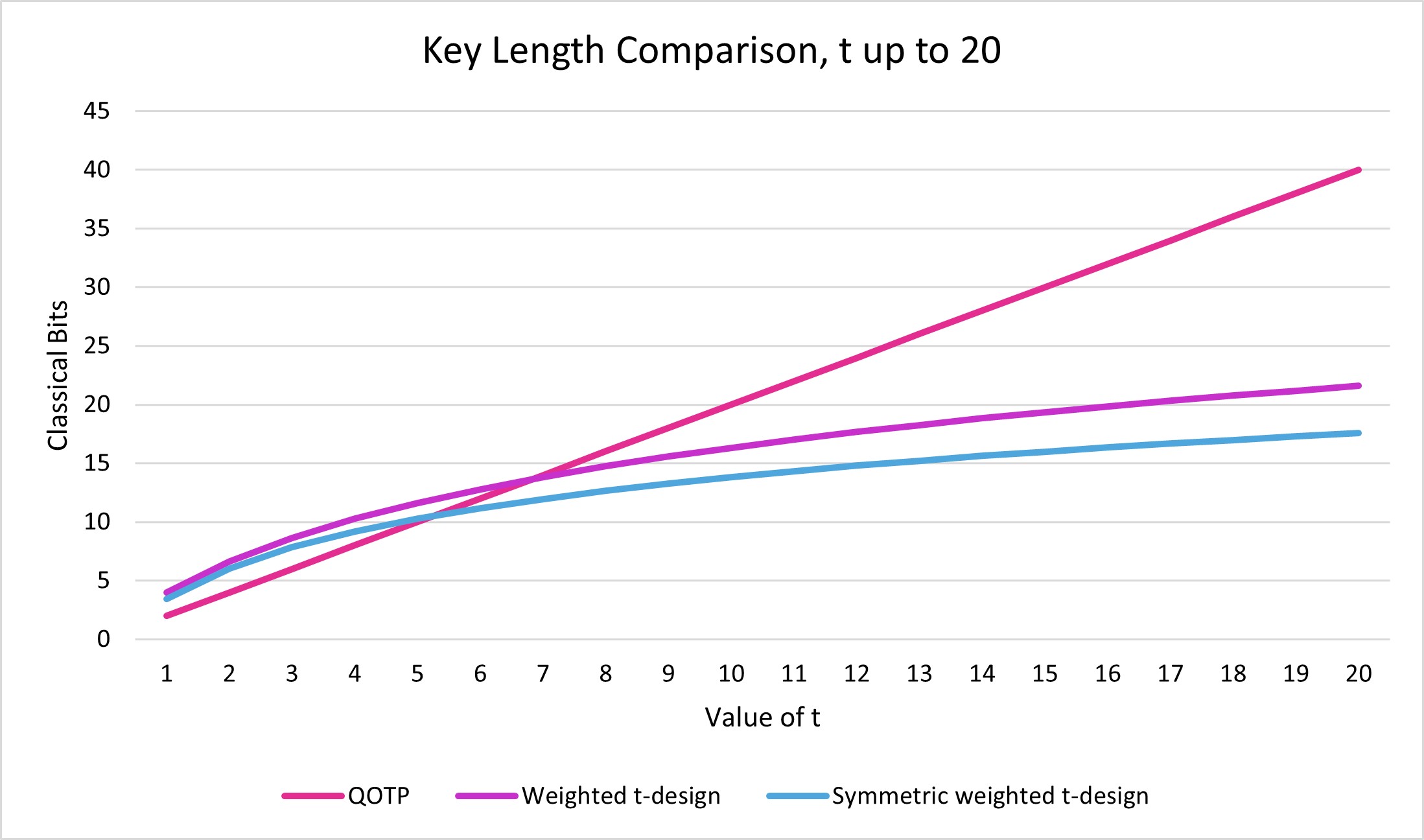}
\caption{
\label{3key20}
QOTP, Weighted $t$-design \& Symmetric Weighted $t$-design, $t\leq 20$, $d=2$}
\label{fig1}
\end{center}
\end{figure}

\subsection{Further Related Work}

Unitary designs were studied as 2-designs~\cite{GAE07,DCEL09} and as  $t$-designs~\cite{AE07}. Follow-up work on $t$-designs includes \cite{RS09,CLLW16,LM20,AMR20,BNOZ20}. Unitary 1-designs are known to yield perfect encryption schemes and unitary 2-designs yield non-malleable encryption schemes~\cite{ABW09} (see also \cite{AM17}). Note that \cite{LM20} considers the  approximate case for unitary 2-designs and their link to approximate non-malleable encryption schemes. Recent work~\cite{BSZ20}  shows that Haar random unitaries  allow a private quantum channel to be implemented with multi-photon pulses, and shows that $t$-designs can be used to practically implement such channels when the parity of the photon source is fixed.

To the best of our knowledge, there do not appear to be efficient constructions of exact unitary $t$-designs for $t>3$, although there has been recent work completed regarding such constructions \cite{BNOZ20,NZO+21}. However, it has been shown that $\epsilon$-approximate unitary $t$-designs on $n$ qubits can be efficiently constructed with local random circuits that are polynomial in $n$, $t$ and $\log(1/\epsilon)$ \cite{BHH16a}. In this work, we use the construction of an $\epsilon$-approximate unitary $t$-design from~\cite{LM20}, where they prove an upper bound for when the unitaries are sampled from an exact $t$-design (Theorem 3.1). They show that when at most $C(td)^t(t\log d)^6/\epsilon^2$ unitaries are sampled from a $t$-design for some constant $C$, then this is an $\epsilon$-approximate unitary $t$-design with probability at least~$\frac{1}{2}$.

We note that recent work on  \emph{private communication over quantum broadcast channels}~\cite{QSW18} considers a different scenario, where recipients are \emph{legitimate} or \emph{malicious}; this differs from our work of broadcasting the same encrypted message to multiple recipients, who must then locally decrypt.

\subsection{Open Problems}
We leave as an open problem further applications of  symmetric unitary $t$-designs, and it would be interesting to see if $t$-designs can be relaxed in similar ways with other subspaces. Relaxing the correctness of the $t$-QPB problem to further improve the key length is left to further research. It is left open whether the techniques used to reduce the circuit depth needed for approximate unitary $t$-designs \cite{HM18,MGDM19,HMMH+20} can be applied to approximate symmetric unitary $t$-designs. We also note that we attain security against side information with $t$-designs by restricting our input of the broadcasting protocol to be in the symmetric subspace, and we leave as an open problem whether there is another solution to $t$-QPB that has the same security and similar key length but with fewer restrictions.

\section{Preliminaries}

In this section, we present the basic notation used throughout this paper. We define unitary $t$-designs, recall the known upper and lower bounds on their size, and briefly define and explain the symmetric subspace and concepts needed from representation theory.

\subsection{Basic Notation}

Let $\Ha_{d^n}$ be the Hilbert space of dimension $d^n$ which is spanned by the basis states $\{\ket{x}: x\in \{0,1,\dots,d-1\}^n\}$. Let $\D(\Ha_{d^n})$ be the set of density operators and $\Lc(\Ha_{d^n})$ be the set of linear operators on~$\Ha_{d^n}$. A Hilbert space of subsystems, say $M$ and $E$, is denoted with subscripts,~$\Ha_{d^n} = \Ha_{M}\otimes \Ha_{E}$. Density operators on such Hilbert spaces are written as $\rho_{ME}$, and~$\rho_E$ denotes when subsystem $M$ is traced out from~$\rho_{ME}$. Transformations between quantum states are formalized by quantum channels, that is, completely positive trace preserving maps. Determining the distinguishability of the outputs from two such channels $\Psi,\Phi: \Lc(\Ha_M) \ra \Lc(\Ha_M)$ is done with the trace norm $||\cdot||_1$, where $||A||_1 = \Tr(\sqrt{AA^{\dagger}})$ for linear operator~$A$. This trace norm is the sum of the singular values of~$A$, while the infinity norm~$||\cdot||_{\infty}$ is the maximum singular value. The quantum channels themselves are compared with the diamond norm $||\cdot||_{\diamond}$, which is the maximum trace norm when an auxiliary space is considered, along with the original Hilbert spaces \cite{Wat11b,BS10}. For example, $||\Psi-\Phi||_{\diamond} = \max_{\rho_{ME}} ||(\Psi\otimes \one_E)\rho_{ME} - (\Phi\otimes\one_{E})\rho_{ME}||_1$. This is considered a better determination of the distinguishability of two quantum channels than the $1\rightarrow 1$ norm, that is, $||\Psi-\Phi||_{1\rightarrow 1} = \max_{\rho_M} ||\Psi (\rho_M)-\Phi(\rho_M)||_1$, because it accounts for the original space $\Ha_M$ being entangled with another auxiliary space~$\Ha_E$.

The notation $\U(d)$ denotes the unitary group of all $d\times d$ unitaries. The Pauli matrices for 2-qubits are defined as
\begin{equation}
\one =
\begin{pmatrix}
1 & 0 \\
0 & 1
\end{pmatrix}
, \;\;\; X =
\begin{pmatrix}
0 & 1 \\
1 & 0
\end{pmatrix}
, \;\;\; Y =
\begin{pmatrix}
0 & -i \\
i & 0
\end{pmatrix}
, \;\;\; Z =
\begin{pmatrix}
1 & 0 \\
0 & -1
\end{pmatrix}.
\end{equation}

The quantum one-time pad (QOTP) is defined in the following way for $\varphi\in\D(\Ha_{2^n})$ and $a,b\in\{0,1\}^n$
\begin{equation}
\mathsf{QOTP}_{a,b} (\varphi) = X^aZ^b\varphi Z^bX^a,
\end{equation}
where $X$ and $Z$ are Pauli operators. The quantum one-time pad is perfectly secure as defined in \cref{defn:QPB-indistinguishable}, where the input need not be restricted to the symmetric subspace.  This is because
\begin{equation}
\Eb_{a,b\in\{0,1\}^n} \mathsf{QOTP}_{a,b}(\varphi) = \frac{\one}{2^{2n}}.
\label{QOTP}
\end{equation}
It can be shown that the quantum one-time pad is also perfectly secure against adversaries with side information (an auxiliary space). This is because $(\Eb_{a,b}\mathsf{QOTP}_{a,b}\otimes\one_E)\varphi_{ME}= \frac{\one}{2^{2n}}\otimes \varphi_E$, where~$\varphi_{ME}\in \D(\Ha_{2^{2n}}\otimes\Ha_E)$.\footnote{This can be generalized to $\Ha_{d^n}$ with the generalized Pauli group.}

\subsection{Unitary $t$-designs}

We use the definition in \cite{RS09} for unitary $t$-designs, which we adapt to our notation.

\begin{dfn}
	Let $\{U_k\}_{k\in K}$ be a finite subset of $\U(d)$ and let $w:\{U_k\}_{k\in K}\ra \R$ be a positive weight function such that $w(U_k)\geq 0$, $\sum_{k\in K}w(U_k)=1$. Then $\mathfrak U =\left(w,\{U_k\}_{k\in K}\right)$ is called a \emph{unitary $t$-design} if
	\begin{equation}\label{design}
		\mathbb E_{\,\mathfrak U} \left[U^{\otimes t} \otimes (U^{\dagger})^{\otimes t}\right] =\sum_{k\in K} w(U_k)\cdot U_k^{\otimes t} \otimes (U_k^{\dagger})^{\otimes t} = \int_{\U(d)} U^{\otimes t} \otimes (U^{\dagger})^{\otimes t} dU,
	\end{equation}
	where the integral is over the whole unitary group with respect to the Haar measure.
	\label{design def}
\end{dfn}

When $w(U_k) = \frac{1}{|K|}$ for every $U_k$, this is an \emph{unweighted} unitary $t$-design. Otherwise, it is a \emph{weighted} unitary $t$-design. The known lower and upper bounds on the number of unitaries needed (\emph{i.e.}~$|K|$) for exact unitary $t$-designs for general $t$ and dimension $d$ are shown in \cref{bounds}.

\begin{table}[H]
\centering
\begin{tabular}{|l|l|l|}
\hline
& \textbf{Lower} & \textbf{Upper} \\ \hline
\textbf{Weighted} & ${{d^2+t-1}\choose{t}}\in \Omega(t^{d^2-1})$ \cite{RS09} & ${{d^2+t-1}\choose{t}}^2\in O(t^{2(d^2-1)})$ \cite{RS09} \\ \hline
\textbf{Unweighted} & ${{d^2+t-1}\choose{t}}\in \Omega(t^{d^2-1})$ \cite{RS09} & $\left(\frac{e(d^2+t-1)}{t}\right)^{2t}$ \cite{AMR20} \\ \hline
\end{tabular}
\caption{Known bounds on the number of unitaries for unitary $t$-designs}
\label{bounds}
\end{table}

There are also approximate unitary $t$-designs, defined as follows.

\begin{dfn}
	Let $\{U_k\}_{k\in K}$ be a finite subset of $\U(d)$ and let $w:\{U_k\}_{k\in K}\ra \R$ be a positive weight function such that $w(U_k)\geq 0$, $\sum_{k\in K}w(U_k)=1$. Then $\mathfrak U =\left(w,\{U_k\}_{k\in K}\right)$ is called an \emph{$\epsilon$-approximate unitary $t$-design} if
\begin{equation}
\left\Vert \mathbb E_{\,\mathfrak U} \left[\mathcal E_{U_k}^{(t)} \right] - T^{(t)} \right\Vert_{ 1\rightarrow 1}  < \epsilon,
\label{approx design}
\end{equation}
where  $T^{(t)}$ is the $t$-twirling channel $T^{(t)}(\rho)=\int_{\U(d)} U^{\otimes t}\rho (U^{\dagger})^{\otimes t} dU$ and $ \mathcal E_{U_k}^{(t)}(\rho)= U_k^{\otimes t}\rho(U_k^{\dagger})^{\otimes t}$ for $\rho \in \D(\Ha_{d^t})$.
\end{dfn}

Note that there are other definitions of $\epsilon$-approximate unitary $t$-designs depending on the norm used in \cref{approx design}. We use the $1\rightarrow 1$ norm as it is the one needed for our application.

\subsection{Symmetric Subspace}
As defined similarly in \cite{Har13}, the symmetric subspace for quantum states in $\Ha_d^{\otimes t}$ is the following:
\begin{equation}
\Sym(d^t) := \{ \ket{\phi} \in (\Ha_d)^{\otimes t} : P_d(\pi)\ket{\phi} = \ket{\phi}, \forall \pi \in S_t \},
\end{equation}
where
\[ P_d(\pi) = \sum_{i_1,\dots,i_t\in [d]} \ketbra{i_{\pi^{-1}(1)},\dots,i_{\pi^{-1}(t)}}{i_1,\dots,i_t}, \]
for $[d] = \{0,\dots,d-1\}$ integers and $\pi\in S_t$, the symmetric group for $t$ elements, \emph{i.e.}~all the permutations of $t$ elements. The dimension for this subspace is $d_{\Sym} = {{d+t-1}\choose{t}}$ \cite{Har13}. The notation $\U(\Sym(d^t))$ denotes unitaries from $\Sym(d^t)\otimes\Sym(d^t)$ of size $d_{\Sym}\times d_{\Sym}$, in the same way that $\U(d)$ denotes unitaries from $\Ha_d\otimes \Ha_d$ of size $d\times d$. The notation $\D(\Sym(d^t))$ is for the density operators on $\Sym(d^t)$. One can write density matrices in the symmetric subspace as a real linear combination of rank $1$ density matrices \cite{Har13}, that is,
\begin{equation}\label{DensitySym}
	\D(\Sym(d^t))\subset \mathrm{span}_{\mathbb R}  \{(\ketbra{\varphi}{\varphi})^{\otimes t}:\ket \varphi\in \Ha_d\}.
\end{equation}

\subsection{Representation Theory}
\label{sec:rep theo}

Using Schur-Weyl duality and Schur's Lemma \cite{FH91} similarly to \cite{LM20}, one can write the following for $\rho\in\D(\Ha_d^{\otimes t})$:
\begin{equation}
\int_{\U(d)} U^{\otimes t} \rho (U^{\dagger})^{\otimes t} dU = \tr(\Pi_{\Sym}\rho \Pi_{\Sym})\tau_{\Sym} + \sum_b \tr(\Pi_b\rho \Pi_b)\tau_b,
\label{SWd}
\end{equation}
where $\Pi_{\Sym}$ is the projector into $\Sym(d^t)$ and $\tau_{\Sym} = \frac{\Pi_{\Sym}}{d_{\Sym}}$. These $\Pi_b$ are projectors into subspaces orthogonal to the symmetric subspace which have dimension $d_b$, and $\tau_b = \frac{\Pi_b}{d_b}$. When $\rho\in\D(\Sym(d^t))$, this reduces to $\tau_{\Sym}$.

\section{Definitions for Quantum Private Broadcasting}
\label{qpb sec}

Here we define the semantics of a \emph{$t$-recipient Quantum Private Broadcast} scheme, ($t$-QPB) along with its security definitions. We also make an observation relating $t$-QPB schemes to $(t-1)$-QPB schemes with perfect security and correctness.
\begin{dfn}
	Let $\Ha_M=\Ha_d$ and $\Ha_C$ be the message and ciphertext Hilbert spaces, respectively. A \emph{$\delta$-correct, $t$-recipient Quantum Private Broadcast scheme in~$\Ha_{M}$} is a set of encryption maps $\mathsf{Enc}_k: \Ha_M^{\otimes t} \ra \Ha_C^{\otimes t}$ along with decryption maps $\mathsf{Dec}_k: \Ha_C \ra \Ha_M$, where  $k \in K$, the set of possible keys. We require that for each $k \in K$,
 $\| \left. (\mathsf{Dec}_k^{\otimes t} \circ \mathsf{Enc}_k)\right|_{\Sym(d^t)}- \one_{\Sym(d^t)}\|_{\diamond} \leq 1-\delta$, where the notation~$\left. \right|_{\Sym(d^t)}$ denotes that the input messages are restricted to being elements of $\Sym(d^t)$, and~$\one_{\Sym}$ is the identity map in $\Sym(d^t)$.

We note that a $1$-correct $t$-QPB, (that is, a perfect $t$-QPB) must necessarily be implemented via unitary matrices. Moreover, in this case, as the definition imposes local identical decryption, the decryption operation needs to be the $t$-fold tensor product of a unitary matrix. Thus, although the encryption maps are not necessarily $t$-fold tensor products of a unitary matrix, the action of each of them over the symmetric subspace can be written as a $t$-fold tensor product of a unitary matrix. Such a perfect $t$-QPB is illustrated in~\cref{QPBfigure}.
\end{dfn}

\begin{figure}[]
\centerline{
\Qcircuit @C=1em @R=1.6em {
&&&& \mbox{Encrypt} &&&& \mbox{Decrypt} &&&& \\
\lstick{} & \ket{\varphi} & & \qw & \gate{U_k} & \qw & \multigate{3}{\text{Send to each recipient}} & \qw & \gate{U_k^{\dagger}} & \qw & \ket{\varphi} \\
\lstick{} & \ket{\varphi} & & \qw & \gate{U_k} & \qw & \ghost{\text{Send to each recipient}} & \qw & \gate{U_k^{\dagger}} & \qw & \ket{\varphi} \\
\lstick{} & \vdots & & & \vdots & & \ghost{\text{Send to each recipient}} & \qw & \vdots & & \vdots \\
\lstick{} & \ket{\varphi} & & \qw & \gate{U_k} & \qw & \ghost{\text{Send to each recipient}} & \qw & \gate{U_k^{\dagger}} & \qw & \ket{\varphi}
\inputgroupv{2}{5}{.8em}{4.6em}{t} \\
}
}
\caption{Quantum Private Broadcasting}
\label{QPBfigure}
\end{figure}
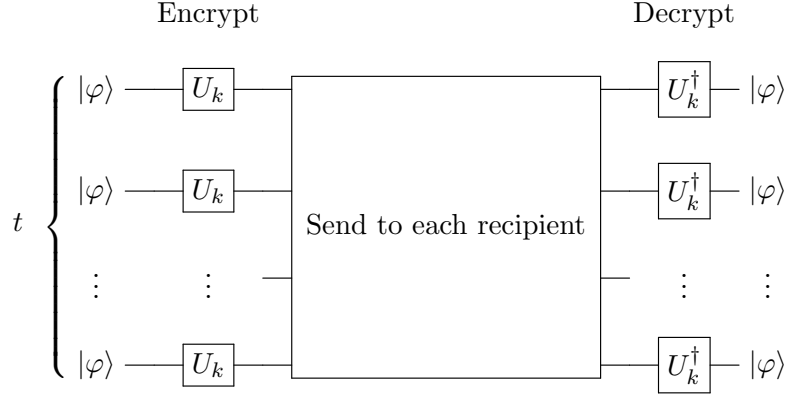
The indistinguishability of ciphertexts for our $t$-QPB encryption scheme is based on the definitions from \cite{LM20}, which compares the encryption scheme with that of a `state replacement channel' $\langle\sigma\rangle$. For a fixed $\sigma \in \D(\Ha_d^{\otimes t})$, this is defined as $\langle\sigma\rangle (R) = \Tr(R)\sigma$, for any $ R\in \D(\Ha_d^{\otimes t})$.

\begin{dfn}
\label{defn:QPB-indistinguishable}
Let $K$ be the set of possible keys in the $t$-QPB. A $t$-QPB has \emph{$\epsilon$-indistinguishable ciphertexts} if there exists a fixed $\sigma \in \D(\Ha_d^{\otimes t})$ such that
\begin{equation}
\left\Vert \left.\left(\Eb_{k\in K} \mathsf{Enc}_k - \langle\sigma\rangle\right)\right|_{\Sym(d^t)} \right\Vert_{1\rightarrow 1} \leq \epsilon.
\end{equation}
We note that the above does not consider quantum side information. The encryption scheme has \emph{$\epsilon$-indistinguishable ciphertexts against adversaries with side information}~if
\begin{equation}
\left\Vert \left.\left(\Eb_{k\in K}  \mathsf{Enc}_k - \langle\sigma\rangle\right)\right|_{\Sym(d^t)} \right\Vert_{\diamond} \leq \epsilon.
\end{equation}
Indistinguishability against adversaries with side information necessarily implies indistinguishability since the $1\ra 1$-norm is upper bounded by the $\diamond$-norm.

When the above norms are equal to zero, we call such encryption schemes \emph{perfectly secure} and \emph{perfectly secure against adversaries with side information}, respectively.
\end{dfn}

When $t=1$, \cref{defn:QPB-indistinguishable} corresponds to the conventional information theoretic encryption~\cite{LM20}, where there is no restriction in the input space.

The following lemma follows naturally from the setting where $t$-copies of a pure quantum state are used as the input for a $t$-QPB.

\begin{lma}
Let $\mathsf{Enc}_k^{(t)}:\Ha_M^{\otimes t}\ra\Ha_C^{\otimes t}$ and $\mathsf{Dec}_k:\Ha_C\ra\Ha_M$, defined as $\mathsf{Enc}_k^{(t)}(\rho) = U_k^{\otimes t}\rho (U_k^{\otimes t})^{\dagger}$ and $\mathsf{Dec}_k(\gamma) = U_k^{\dagger}\gamma U_k$, respectively. Let $(\mathsf{Enc}_k^{(t)},\mathsf{Dec}_k)$ be a perfectly secure and perfectly correct $t$-QPB scheme. Then $(\mathsf{Enc}_k^{(t-1)},\mathsf{Dec}_k)$ is a perfectly secure and perfectly correct $(t-1)$-QPB scheme.\label{QPBt}
\end{lma}

\begin{proof}
By definition of encoding and decoding maps it is clear that for any $\rho \in \D(\Ha_d^{\otimes t-1}\otimes \Ha_A)$, we have $\left(\mathsf{Dec}_k^{\otimes t-1} \circ \mathsf{Enc}^{(t-1)}_k\right)\otimes \id_A(\rho)=\rho$ and thus \newline $ \left\Vert \left. \left(\mathsf{Dec}_k^{\otimes t-1} \circ \mathsf{Enc}_k^{(t-1)}\right)\right|_{\Sym(d^{t-1})}- \one_{\Sym(d^{t-1})}\right\Vert_{\diamond} = 0$ showing correctness.

Let $\rho = (\ketbra{\varphi}{\varphi})^{\otimes t-1}$, with $\ket \varphi\in \Ha_d$, then we have
$$\Eb_{k\in K} \mathsf{Enc}^{(t-1)}_k (\rho) =\tr_1\left(\Eb_{k\in K} \mathsf{Enc}^{(t)}_k \left(\rho\otimes \ketbra{\varphi}{\varphi}\right) \right) =\tr_1\left(\tau_{\Sym,t}\right)=\tau_{\Sym,t-1},$$ where the first equality follows from linearity, and the second follows from the definition of a perfectly correct $t$-QPB scheme. We use the notation $\tau_{\Sym,t}$ to make explicit that it is the maximally mixed state in $\D(\Sym(d^t))$. Moreover, using \cref{DensitySym} and linearity, we know that this equation holds for any $\rho\in\D(\Sym(d^t))$. Thus,
\begin{equation*}
	\left\Vert \left.\left(\Eb_{k\in K} \mathsf{Enc}^{(t-1)}_k - \langle\tau_{\Sym,t-1} \rangle\right)\right|_{\Sym(d^{t-1})} \right\Vert_{1\rightarrow 1} =0.
	\end{equation*}
\end{proof}

\section{Limitations on the Quantum One-Time Pad}
\label{qotp section}

When considering classical encryption, the one-time pad (OTP) can only be used once to encrypt a plaintext message since the exclusive-or (XOR) of the ciphertexts resulting from encrypting different plaintexts reveals information about these plaintexts. However, if the OTP is used to encrypt two (or more) identical plaintexts, their ciphertexts will also be identical and the XOR of these ciphertexts is the zero string. This reveals nothing about the original plaintext, and therefore is still information-theoretically secure.

Since classical messages are a special case of quantum messages, the QOTP should also only be used once to encrypt a plaintext quantum state for the same reasons as the OTP. However, when the QOTP is used to encrypt two copies of the same quantum state, this is no longer information-theoretically secure, as illustrated in the following theorem.

\begin{thm}
 $\mathsf{QOTP}_{a,b}\otimes\mathsf{QOTP}_{a,b}$ with the same key $a,b$ is a 1-correct, 2-recipient QPB scheme, but it does not have $\epsilon$-indistinguishable ciphertexts for any $\epsilon < 1/2$.
\end{thm}

\begin{proof}
This $\mathsf{QOTP}_{a,b}\otimes\mathsf{QOTP}_{a,b}$ can be defined as a ``double quantum one-time pad" for $\varphi, \psi \in\D(\Ha_{2})$ and $a,b\in\{0,1\}$:
\begin{equation}
\mathsf{dQOTP}_{a,b}(\varphi\otimes \psi) = X^aZ^b\otimes X^aZ^b(\varphi \otimes \psi) Z^bX^a \otimes Z^bX^a.
\end{equation}
Consider the following:
\begin{equation}
\begin{split}
\rho_0 & = \ketbra{0}{0} \otimes \ketbra{0}{0} \\
\rho_1 & = \ketbra{+}{+} \otimes \ketbra{+}{+}\,.
\end{split}
\end{equation}
Then the expectation of $\mathsf{dQOTP}_{a,b}$ applied to each state results in
\begin{equation}
\begin{split}
\Eb_{a,b}\mathsf{dQOTP}_{a,b}(\rho_0) & = \frac{1}{2}\Big(\ketbra{0}{0}\otimes \ketbra{0}{0} + \ketbra{1}{1} \otimes \ketbra{1}{1} \Big) \\
\Eb_{a,b}\mathsf{dQOTP}_{a,b}(\rho_1) & = \frac{1}{2} \Big(\ketbra{+}{+}\otimes \ketbra{+}{+} + \ketbra{-}{-}\otimes \ketbra{-}{-} \Big).
\end{split}
\end{equation}
We have that, for any state replacement channel $\langle \sigma \rangle$,
	\begin{equation*}
		\begin{aligned}
		&\left\Vert \left.\Big(\Eb_{a,b}\mathsf{dQOTP}_{a,b} - \langle\sigma\rangle\Big)\right|_{\Sym(2^2)} \right\Vert_{1\rightarrow 1} =\max_{\rho\in \D(\Sym(2^2))} \left\Vert \Eb_{a,b}\mathsf{dQOTP}_{a,b}(\rho) - \langle\sigma\rangle(\rho) \right\Vert_{1}  \\
		&\geq \frac 1 2 \left(\left\Vert \Eb_{a,b}\mathsf{dQOTP}_{a,b}(\rho_0) - \langle\sigma\rangle(\rho_0) \right\Vert_{1} + \left\Vert \Eb_{a,b}\mathsf{dQOTP}_{a,b}(\rho_1) - \langle\sigma\rangle(\rho_1) \right\Vert_{1}\right) \\
		&\geq \frac 1 2 \left\Vert \Eb_{a,b}\mathsf{dQOTP}_{a,b}(\rho_0)-\Eb_{a,b}\mathsf{dQOTP}_{a,b}(\rho_1)\right\Vert_1 \geq \frac{1}{2}. \qedhere
		\end{aligned}
	\end{equation*}
\end{proof}
Therefore, encryption with the same key is not sufficient to obtain perfect security when encrypting multiple copies of the same message. Using independent encryption keys for each copy of the message is one possible solution to this problem. However, as one can see in  \cref{qotp bounds}, this leads to the amount of unitaries needed to be exponential in $t$, the number of copies. These bounds are from the known fact that to encrypt once an $n$-qubit state, $2^{2n}$ unitaries are needed \cite{BR03}, and this bound can be extended to general $d$ with a general QOTP using generalized Pauli matrices \cite{Web16}. We denote $t$ independent uses of the quantum one-time pad as $\mathsf{QOTP}_{a_i,b_i}^{\otimes t}$, where $a_i,b_i \in \{0,1\}^n$ for $i=1,\dots,t$. In the $d$-dimensional case, $a_i,b_i \in \{0,1,\dots,d-1\}^n$.

\begin{table}[H]
\centering
\begin{tabular}{|l|l|l|}
\hline
& $\mathsf{QOTP}_{a,b}$ & $\mathsf{QOTP}_{a_i,b_i}^{\otimes t}$ \\ \hline
\textbf{Qubits $(d=2^n)$} & $d^2 = 4^n$ & $d^{2t} = 4^{nt}$ \\ \hline
\textbf{General $d^n$} & $d^{2n}$ & $d^{2nt}$ \\ \hline
\end{tabular}
\caption{Bounds on the number of unitaries for Quantum One-Time Pad for $n$~qudits}
\label{qotp bounds}
\end{table}

\section{QPB with Designs}

In this section we examine the case where unitary $t$-designs are used to solve the $t$-QPB problem. In order to maintain security against side information, we impose restrictions on the input message, specifically that it be an element of the symmetric subspace.

\begin{thm}\label{secureQPB}
Let $\mathfrak U=(w,\{U_k\}_{k\in K})$ be an $\epsilon$-approximate unitary $t$-design. Then the set of maps $\mathsf{Enc}_k(\rho)=U_k^{\otimes t}\rho(U_k^{\otimes t})^ \dagger$ and its local inverse maps $\mathsf{Dec}_k(\gamma)=U_k^ \dagger \gamma U_k$ for $k \in K$, $\rho\in\D(\Sym(d^t))$, and $\gamma\in\D(\Ha_d)$ form a perfect $t$-QPB which has $\epsilon$-indistinguishable ciphertexts. Moreover, in the case of exact unitary $t$-designs, we have a perfect $t$-QPB perfectly secure against adversaries with side information.
\end{thm}

\begin{proof}
The fact that $\mathsf{Enc}_k$ and $\mathsf{Dec}_k^{\otimes t}$ are inverses of each other automatically shows correctness. Denote  $T^{(t)}$ the $t$-twirling channel $T^{(t)}(\rho)=\int_{\U(d)} U^{\otimes t}\rho (U^{\dagger})^{\otimes t} dU$. For $\rho\in\D(\Sym(d^t))$, $T^{(t)}(\rho)=\tau_{\Sym}$, that is, $T^{(t)}|_{\Sym(d^t)}=\langle\tau_{\Sym}\rangle|_{\Sym(d^t)}$, thus using the definition of approximate $t$-designs we get

\begin{equation*}
	\begin{aligned}
	\left\Vert \left.\left(\Eb_{k\in K} \mathsf{Enc}_k - \langle\tau_{\Sym}\rangle\right)\right|_{\Sym(d^t)} \right\Vert_{1\rightarrow 1}
	&=\left\Vert \left.\left( \Eb_{k\in K} \mathsf{Enc}_k - T^{(t)}\right)\right|_{\Sym(d^t)}  \right\Vert_{ 1\rightarrow 1}\\
	&\leq \left\Vert \Eb_{k\in K} \mathsf{Enc}_k  - T^{(t)} \right\Vert_{ 1\rightarrow 1}< \epsilon.
\end{aligned}
\end{equation*}

Consider now the security against side information for the case of exact unitary $t$-designs. Suppose the plaintext to be encrypted is $\ket{\psi}\in\Ha_A \otimes \Sym(d^t)$, where $A$ is the auxiliary space. This can be written as
\begin{equation}
\begin{split}
\ket{\psi} & = \sum_{i=1}^{D} \lambda_i \ket{a_i} \otimes \ket{\varphi_i} \\
\ketbra{\psi}{\psi} & = \sum_i \sum_j \lambda_i \lambda_j^* \ketbra{a_i}{a_j} \otimes \ketbra{\varphi_i}{\varphi_j},
\end{split}
\end{equation}
using the Schmidt decomposition, where $\ket{a_i}$ and $\ket{\varphi_i}$ are orthonormal states for $\Ha_A$ and $\Sym(d^t)$, respectively. The $\lambda_i$ values are non-negative real numbers such that $\sum_i \lambda_i^2 = 1$.

Applying $\one_A\otimes\mathsf{Enc}_k$ to $\ketbra{\psi}{\psi}$ and taking the expectation gives
\begin{equation}
\begin{split}
& \sum_i \sum_j \lambda_i \lambda_j^* \ketbra{a_i}{a_j} \otimes \sum_{k\in K} w(U_k) U_k^{\otimes t} \ketbra{\varphi_i}{\varphi_j} (U_k^{\dagger})^{\otimes t} \\
& = \sum_i \sum_j \lambda_i \lambda_j^* \ketbra{a_i}{a_j} \otimes \int_{\U(d)} U^{\otimes t} \ketbra{\varphi_i}{\varphi_j} (U^{\dagger})^{\otimes t} dU \\
& = \sum_i \sum_j \lambda_i \lambda_j^* \ketbra{a_i}{a_j} \otimes \Big(\tr(\Pi_{\Sym}\ketbra{\varphi_i}{\varphi_j} \Pi_{\Sym})\tau_{\Sym} + \sum_b \tr(\Pi_b\ketbra{\varphi_i}{\varphi_j} \Pi_b)\tau_b\Big).
\end{split}
\end{equation}
The second equality follows from \cref{SWd}, whose notation is explained in \cref{sec:rep theo}.
This $\tr(\Pi_{\Sym} \ketbra{\varphi_i}{\varphi_j} \Pi_{\Sym}) = \mel{\varphi_j}{\Pi_{\Sym}\Pi_{\Sym}}{\varphi_i}$ will equal 0 when $i\neq j$ since $\ket{\varphi_i}$ and $\ket{\varphi_j}$ are orthonormal. For $\tr(\Pi_b \ketbra{\varphi_i}{\varphi_j}\Pi_b)$, this will always equal zero because $\ket{\varphi_i}, \ket{\varphi_j} \in \Sym(d^t)$ which is orthogonal to subspace~$b$, and so $\Pi_b$ applied to these states will give zero. Therefore the only terms that remain are when $i=j$, which gives
\begin{equation}
\begin{split}
& \sum_i |\lambda_i|^2 \ketbra{a_i}{a_i} \otimes \int_{\U(d)} U^{\otimes t} \ketbra{\varphi_i}{\varphi_i} (U^{\dagger})^{\otimes t} dU \\
& = \sum_i |\lambda_i|^2 \ketbra{a_i}{a_i} \otimes \tau_{\Sym},
\end{split}
\end{equation}
and this $\tau_{\Sym}$ is independent of $i$. This implies that the encrypted plaintext will always look the same, regardless of what the adversary has as side information.
This implies
 \begin{equation*}
	\left\Vert \Big(\Eb_{k\in K} \left.\mathsf{Enc}_k - \langle\tau_{\Sym}\rangle\Big) \right|_{\Sym(d^t)}\right\Vert_{\diamond} = 0.
\end{equation*}
\end{proof}

\begin{rmk}
	Quantum Private Broadcasting with designs for $t$-recipients cannot be used to broadcast states of the form $\nu^{\otimes t}\notin \D(\Sym(d^t))$.  Consider for example, the totally mixed state $\tau=\frac{\one}{2}\otimes \frac{\one}{2} \in\D(\Ha_{d^t})$ for $d=t=2$.
	The averaged encryption of $\tau$ is naturally
	$\mathbb{E}_{k\in K} \mathsf{Enc}_k(\tau)=\tau$. On the other hand, any state $\rho_0\in \D(\Sym(2^2))$ is mapped to $\mathbb{E}_{k\in K}\mathsf{Enc}_k(\rho_0)=\tau_\Sym$, the maximally mixed state in the symmetric subspace. Clearly $\frac{\one}{4} \neq \tau_{\Sym}$ because when $d=t=2$,
\begin{equation}
\begin{split}
\tau_{\Sym} & = \frac{\Pi_{\Sym}}{d_{\Sym}} \\
& = \frac{\Pi_{\Sym}}{3} \neq \frac{\one}{4},
\end{split}
\end{equation}	
and for any state replacement channel $\langle \sigma \rangle$,
	\begin{equation*}
		\begin{aligned}
		\left\Vert \Big(\Eb_{k\in K}\mathsf{Enc}_{k} - \langle\sigma\rangle\Big) \right\Vert_{1\rightarrow 1} &=\max_{\rho\in \D(\C(2^2))} \left\Vert \Eb_{k\in K}\mathsf{Enc}_{k}(\rho) - \langle\sigma\rangle(\rho) \right\Vert_{1}  \\
		&\geq \frac 1 2 \left(\left\Vert \Eb_{k\in K}\mathsf{Enc}_{k}(\tau) - \langle\sigma\rangle(\tau) \right\Vert_{1} + \left\Vert \Eb_{k\in K}\mathsf{Enc}_{k}(\rho_0) - \langle\sigma\rangle(\rho_0) \right\Vert_{1}\right) \\
		&\geq \frac{1}{2} \left(\left\Vert \Eb_{k\in K}\mathsf{Enc}_k(\tau) - \Eb_{k\in K}\mathsf{Enc}_k(\rho_0)\right\Vert_1\right) \\
		&\geq \frac 1 2 \left\Vert \tau - \tau_{\Sym}\right\Vert_1 \geq \frac{1}{4}.
		\end{aligned}
	\end{equation*}
	This does not fulfill a generalized version of security following \cref{defn:QPB-indistinguishable}, and therefore supports why we restrict our input to the symmetric subspace in the definitions of security and correctness for $t$-QPB. Furthermore, we can insert a pre-broadcasting stage into the $t$-QPB where we perform a projective measurement $\{\tau_{\Sym}, \one-\tau_{\Sym}\}$ to determine whether or not our state is in the symmetric subspace. The state provided by an adversary is either projected into the symmetric subspace, whose action leaves symmetric states unchanged, or it is projected into a subspace orthogonal to the symmetric subspace. In the first case, the state is symmetric and the $t$-QPB is secure, as explained above. In the second case, the projective measurement result indicates that the state is not symmetric, and the encryption protocol is aborted, thus avoiding scenarios where the $t$-QPB is not secure.
\end{rmk}

We are interested in the key length required for the $t$-QPB, and we can compare the bounds from \cref{qotp bounds} to those in \cref{bounds}. One can see that the upper bounds for unitary $t$-designs are better than the number of unitaries needed for $\mathsf{QOTP}_{a_i,b_i}^{\otimes t}$ when $t$ is very large. The reason for this is because if one fixes the dimension $d$ and allows $t$ to increase, the order of unitaries needed for a $t$-design is polynomial in $t$, while the QOTP is exponential in $t$.
See \cref{3key20} for the comparison of the classical bit key length when $d=2$ and $t=1,\dots,20$.

\section{Symmetric Unitary $t$-designs}
\label{sym sec}

Motivated by the fact that we are only working in the symmetric subspace, we propose the new concept of symmetric unitary $t$-designs, which are a relaxation of $t$-designs. Namely, they are a discrete set of unitaries together with a probability distribution that mimics the action of the Haar measure \emph{in the symmetric subspace}.

\begin{dfn}
Let $\{U_k\}_{k\in K}$ be a finite subset of $\U(d)$ and let $w:\{U_k\}_{k\in K}\ra \R$ be a positive weight function such that $w(U_k)\geq 0$ and $\sum_{k\in K}w(U_k)=1$. Then $\mathfrak U =\left(w,\{U_k\}_{k\in K}\right)$ is called an \emph{$\epsilon$-approximate symmetric unitary $t$-design} if
	\begin{equation}
			\left \Vert\left.\left(\mathbb{E}_{\mathfrak{U}}\left[ \mathcal E_{U_k}^{(t)}\right]-\left\langle \tau_{\Sym}\right\rangle \right) \right|_{\Sym(d^t)}\right\Vert_{1\rightarrow 1 }<\epsilon ,
		\end{equation}
where $ \mathcal E_{U_k}^{(t)}(\rho)=U_k^{\otimes t}\rho (U_k^\dagger)^{\otimes t}$.
\label{sym design}
\end{dfn}

Note that $\left\langle \tau_{\Sym}\right\rangle$ is equal to $T^{(t)}$, the $t$-twirling channel $T^{(t)}(\rho)=\int_{\U(d)}  U^{\otimes t}\rho (U^\dagger)^{\otimes t}d U$, for symmetric states $\rho\in \D(\Sym(d^t))$ and the integral is over the whole unitary group with respect to the Haar measure.

We now connect symmetric unitary $t$-designs with perfect $t$-QPB schemes.

\begin{cor}
	Let $\mathfrak U=(w,\{U_k\}_{k\in K})$ be an $\epsilon$-approximate symmetric unitary $t$-design. Then the set of maps $\mathsf{Enc}_k(\rho)=U_k^{\otimes t}\rho(U_k^{\otimes t})^ \dagger$ and its local inverse maps $\mathsf{Dec}_k(\gamma)=U_k^ \dagger \gamma U_k$ for $k \in K$, $\rho\in\D(\Sym(d^t))$, and $\gamma\in\D(\Ha_d)$ form a perfect $t$-QPB which has $\epsilon$-indistinguishable ciphertexts. Moreover, in the case of exact symmetric unitary $t$-designs we have a perfect $t$-QPB perfectly secure against adversaries with side information.
\end{cor}

\begin{proof}
Note that the only properties of approximate or exact unitary $t$-designs we are using in the proof of \cref{secureQPB} is the one fulfilled by their correspondent symmetric unitary $t$-designs.
\end{proof}

This shows that symmetric unitary $t$-designs give perfect $t$-QPB schemes. Moreover, every perfect $t$-QPB comes from a symmetric unitary $t$-design. Indeed, as discussed after the definition of $t$-QPB schemes in \cref{qpb sec}, perfect $t$-QPB must necessarily be implemented via unitary matrices and with local identical decryption unitaries $U_k$. Encryption can be performed with a general unitary for each $U_k$, but its action over the symmetric subspace should be exactly the same as $(U_k^\dagger)^{\otimes t}$. So mathematically, the $t$-QPB comes from a symmetric unitary $t$-design.

Hence, \cref{QPBt} can be rephrased in terms of symmetric unitary $t$-designs.
\begin{lma}
	Let $\mathfrak U =\left(w,\{U_k\}_{k\in K}\right)$ be a symmetric unitary $t$-design, then $\mathfrak U $ is a symmetric unitary $(t-1)$-design.
\end{lma}

We now give lower and upper bounds for exact symmetric unitary $t$-designs.
\begin{lma}
A symmetric unitary $t$-design has at least $d_{\Sym}^2$ unitaries.
\end{lma}
\begin{proof}
A symmetric $t$-design in $\U(d)$ gives a $1$-design in~$ \U(\Sym(d^t))$ having a particular tensor product structure, via the map $U\in \U(d) \mapsto V_U=U^{\otimes t}|_{\Sym(d^t)} \in \U(\Sym(d^t))$, where $U^{\otimes t}|_{\Sym(d^t)}: \Sym(d^t)\rightarrow \Sym(d^t)$ is the restriction of $U^{\otimes t}$ to the symmetric subspace.

Therefore, a lower bound for the number of unitaries needed in a 1-design in $\U(\Sym(d^t))$ will also give a lower bound for those of a symmetric $t$-design in $\U(d)$. From \cref{bounds}, the lower bound for a $t$-design in $\U(d)$ is ${{d^2+t-1}\choose{t}}$. This implies that the lower bound on the number of unitaries for a symmetric 1-design is ${{d_{\Sym}^2+1-1}\choose{1}} = d_{\Sym}^2$. \looseness=-1
\end{proof}

\begin{lma}
	There are symmetric unitary $t$-designs formed by $n$ unitaries with $n\leq d_{\Sym}^4-2d_{\Sym}+3$ unitaries.
\end{lma}

\begin{proof}
The proof follows using the results from \cite{RS09} regarding the dimensions for sets of homogeneous polynomials and then applying Carath\'eodory's theorem.

A symmetric unitary design seen as a linear operator is an element of the convex hull of the set
\begin{equation}
	A = \{U^{\otimes t}\otimes (\overline U)^{\otimes t}|_{\Sym(d^t)\otimes \Sym(d^t)} : U\in\U(d)\}.
\end{equation}
Clearly, the convex hull of $A$, is a subset of the convex hull of $B = \{V\otimes \overline V : V\in\U(\Sym(d^t))\}$, where $V$ does not necessarily have the tensor product structure. The span of set $B$ has the same dimension as  $\Hom(\U(\Sym(d^t)),1,1)$, the set of homogenous polynomials of degree $1$  in the entries of $V$ and degree $1$ in the entries of $\overline V$ where $V\in\U(\Sym(d^t))$, whose dimension is $d_{\Sym}^4-2d_{\Sym}+2$ (see \cite{RS09}). Now applying Carath\'eodory's theorem, elements of the convex hull of $A$ can be written as convex combinations of at most $d_{\Sym}^4-2d_{\Sym}+3$ elements in $A$. Therefore, there exists a weighted symmetric unitary $t$-design of at most  $d_{\Sym}^4-2d_{\Sym}+3 \in O(d_{\Sym}^4)$ elements.
\end{proof}

This shows a gap between the lower and upper bounds in line with the results for unitary $t$-designs. The bounds are summarized in  \cref{sym bounds}.

We concentrate now in giving bounds on the number of unitaries needed for approximate symmetric unitary $t$-designs. We adapt the randomized construction of approximate unitary $t$-designs from \cite{LM20} \footnote{These results build on those of \cite{Aub09}. Note that we make explicit this $\log(1/\epsilon^2)$ term that is missing in the result of \cite{Aub09}.} to our case where we are only interested in the action of the set of unitary matrices over $\Sym(d^t)$, giving a construction almost linear in $d_{\Sym}$.

\begin{thm}
Let $0<\epsilon<1$. Let $\mathfrak U=(w,\{U_k\}_{k\in K})$ be a unitary $t$-design,  and let $U_1,\dots,U_n$ be sampled independently from $\mathfrak U$. Then there exists a universal constant $\alpha>0$ such that, if $n\geq \alpha \frac{d_{\Sym}}{\epsilon^2}\log(d_{\Sym})^6 \log(1/\epsilon^2)$ then with probability at least $\frac{1}{2}$,
\begin{equation}
\forall \; \rho \in \D(\Sym(d^t)), \;\; \left\Vert \frac{1}{n} \sum_{i=1}^n U_i^{\otimes t} \rho (U_i^{\dagger})^{\otimes t} - T^{(t)}(\rho) \right\Vert_{\infty} \; \leq \frac{\epsilon}{d_{\Sym}},
\end{equation}
where $T^{(t)}(\rho)$ is the symmetric $t$-twirling channel which maps $\rho\in\D(\Sym(d^t))$ to $\int_{\U(d)} U^{\otimes t}\rho (U^{\dagger})^{\otimes t} dU$ with respect to the Haar measure. In other words, $T^{(t)}(\rho)=\langle \tau_{\Sym}\rangle (\rho)=\tau_{\Sym}$ for $\rho \in \D(\Sym(d^t))$.
\label{Cec alt}
\end{thm}

Note that, by relating the $\infty$ norm to the $1$ norm, \cref{Cec alt} gives an $\epsilon$-approximate symmetric unitary $t$-design and thus a perfectly correct $t$-QPB scheme which has $\epsilon$-indistinguishable ciphertexts.

The proof of \cref{Cec alt} follows similarly to \cite{LM20}, with altered bounds due to $\rho$ being in the symmetric subspace. To see this, we need the following result based from \cite{Aub09} Lemma 5, now adjusted so that $\rho\in\D(\Sym(d^t))$ and $U_i^{\otimes t}$ is being applied instead of simply $U_i$.

\begin{lma}
Let $U_1,\dots,U_n \in \U(d)$. For $\varepsilon_1,\dots,\varepsilon_n$ independent Bernoulli random variables, we have
\begin{equation}
\begin{split}
&\Eb \left(\sup_{\rho\in\D(\Sym(d^t))} \left\Vert \sum_{i=1}^n \varepsilon_i U_i^{\otimes t} \rho (U_i^{\dagger})^{\otimes t} \right\Vert_{\infty} \right) \\
& \;\;\;\; \leq \alpha(\log d_{\Sym})^{5/2} (\log n)^{1/2} \; \sup_{\rho\in\D(\Sym(d^t))} \left\Vert \sum_{i=1}^n U_i^{\otimes t} \rho (U_i^{\dagger})^{\otimes t} \right\Vert_{\infty}^{1/2},
\end{split}
\end{equation}
where $\alpha>0$ is a universal constant.
\label{repl Cec lem}
\end{lma}

\begin{proof}
This proof follows from the proof in \cite{Aub09} since there exists an isometry that will map everything in $\Sym(d^t)$ to a complex Hilbert space $\Ha_{d_{\Sym}}$ of $d_{\Sym}$ dimensions. This isometry preserves scalar products and maps all non-symmetric elements to zero. There is therefore an isometry between $\D(\Sym(d^t))$ and $\D(\Ha_{d_{\Sym}})$, and Aubrun's Lemma 5 result can be applied, where $d$ is replaced with $d_{\Sym}$ and $U_i$ is now $U_i^{\otimes t}$.
\end{proof}

\cref{Cec alt} can now be proved by directly following the proof of \cite{LM20}, replacing Lemma 3.2 with the known fact of $\sup_{\rho\in\D(\Sym(d^t))}\left\Vert T^{(t)}(\rho)\right\Vert_{\infty} = \frac{1}{d_{\Sym}}$ and substituting \cref{repl Cec lem} for Lemma~3.3 in~\cite{LM20}.

From \cite{LM20}, their upper bound is $n\geq C(td)^t(t\log d)^6 /\epsilon^2$, while the upper bound from \cref{Cec alt} is $n\geq \alpha \frac{d_{\Sym}}{\epsilon^2}\log(d_{\Sym})^6 \log(1/\epsilon^2)$. As mentioned previously, the lower bound for symmetric unitary $t$-designs is $d_{\Sym}^2$, and this upper bound for $\epsilon$-approximate symmetric unitary $t$-designs is of order $d_{\Sym}$ along with a $\log d_{\Sym}$ term. The following lemma shows that this upper bound is optimal in $d_{\Sym}$ up to a sublinear term.

\begin{table}[H]
\centering
\begin{tabular}{|l|l|l|}
\hline
& \textbf{Lower} & \textbf{Upper} \\ \hline
\textbf{Exact} & $d_{\Sym}^2$ & $d_{\Sym}^4-2d_{\Sym}^2+3 \in O(d_{\Sym}^4)$ \\ \hline
\textbf{$\epsilon$-Approximate} & $(d_{\Sym})^{(1-\epsilon)}$ & $\alpha \frac{d_{\Sym}}{\epsilon^2}\log(d_{\Sym})^6 \log(1/\epsilon^2)$ \\ \hline
\end{tabular}
\caption{Bounds on the number of unitaries for Symmetric Unitary $t$-designs}
\label{sym bounds}
\end{table}

\begin{lma}
	An $\epsilon$-approximate symmetric unitary $t$-design has at least $(d_{\Sym})^{1-\epsilon}$ unitaries.
\end{lma}

\begin{proof}
We adapt the arguments given in \cite{LW17} to our case. As proven in \cite{LW17}, if two quantum channels $T$ and $\hat{T}$ on $\Lc(\Ha_d)$ are $\epsilon$-close in the 1-norm, then the following is true
\begin{equation}
\log r(\hat{T}) \geq (1-\epsilon) \max_{\rho\in \D(\Ha_d)} \left|S(T(\rho)) - S(\rho) \right|,
\end{equation}
where $r(\hat{T})$ is the Kraus rank of $\hat{T}$ and $S(\cdot)$ is the von Neumann entropy.

In \cite{LM20} it is explained that if the quantum channel $T$ has the property that $\left\Vert T(\rho)\right\Vert_{\infty} \leq \frac{c}{d}$ for $\rho\in\D(\Ha_d)$, then it can be said that
\begin{equation}
\max_{\rho\in\D(\Ha_d)} \left| S(T(\rho)) - S(\rho) \right| \geq \log \left(\frac{d}{c}\right),
\end{equation}
which implies that $r(\hat{T}) \geq \left(\frac{d}{c}\right)^{(1-\epsilon)}$.

With respect to approximate symmetric unitary $t$-designs, it is known that for $\rho\in\D(\Sym(d^t))$, $\left\Vert T^{(t)}(\rho)\right\Vert_{\infty} = \frac{1}{d_{\Sym}}$. Therefore, if a quantum channel $\hat{T}^{(t)}$ is $\epsilon$-close to $T^{(t)}$ in the 1-norm, then the rank of Kraus operators for the channel $\hat{T}^{(t)}$ satisfies
\begin{equation}
r(\hat{T}^{(t)}) \geq (d_{\Sym})^{(1-\epsilon)},
\end{equation}
which gives a lower bound for the number of unitaries needed for an $\epsilon$-approximate symmetric unitary $t$-design.
\end{proof}

\section*{Acknowledgements}

The authors thank Cécilia Lancien for helpful discussions. A.B.~acknowledges support by the Air Force Office of Scientific Research under award number FA9550-20-1-0375, Canada's NFRF, Canada's NSERC, an Ontario ERA, and the University of Ottawa's Research Chairs program. C.G.~acknowledges financial support from Spanish MICINN (project MTM2017-88385-P) and from Comunidad de Madrid (grant QUITEMAD-CM, ref. S2018/TCS-4342). C.S.~acknowledges financial support from the Government of Ontario and from the University of Ottawa.

\appendix

\section{Data for \cref{3key20}}
\label{sec:appendix-data}

\begin{table}[H]
\centering
\begin{tabular}{|l|l|l|l|}
\hline
\textbf{$t$}& \textbf{QOTP} & \textbf{Weighted $t$-design} & \textbf{Symmetric Weighted $t$-design} \\ \hline
1 & 2 & 4 & 3.46 \\ \hline
2 & 4 & 6.64 & 6.04 \\ \hline
3 & 6 & 8.64 & 7.83 \\ \hline
4 & 8 & 10.26 & 9.17 \\ \hline
5 & 10 & 11.61 & 10.26 \\ \hline
6 & 12 & 12.78 & 11.17 \\ \hline
7 & 14 & 13.81 & 11.96 \\ \hline
8 & 16 & 14.73 & 12.64 \\ \hline
9 & 18 & 15.56 & 13.26 \\ \hline
10 & 20 & 16.32 & 13.81 \\ \hline
11 & 22 & 17.02 & 14.32 \\ \hline
12 & 24 & 17.66 & 14.78 \\ \hline
13 & 26 & 18.26 & 15.21 \\ \hline
14 & 28 & 18.82 & 15.61 \\ \hline
15 & 30 & 19.34 & 15.99 \\ \hline
16 & 32 & 19.84 & 16.34 \\ \hline
17 & 34 & 20.31 & 16.67 \\ \hline
18 & 36 & 20.75 & 16.98 \\ \hline
19 & 38 & 21.18 & 17.28 \\ \hline
20 & 40 & 21.58 & 17.56 \\ \hline
\end{tabular}
\caption{Classical bits for QOTP and upper bounds of classical bits for weighted $t$-design, symmetric weighted $t$-design when $d=2$}
\label{data:bits}
\end{table}

\bibliographystyle{bibtex/bst/alphaarxiv}
\bibliography{bibtex/bib/full,bibtex/bib/quantum}

\end{document}